\documentclass[11pt,onecolumn,a4paper]{article}

\usepackage[left=2.2cm,right=2.2cm,top=2.4cm,bottom=2.4cm]{geometry}
\usepackage[affil-it]{authblk}
\usepackage[utf8]{inputenc}
\usepackage{amsmath,amssymb,amsfonts,amsthm}
\usepackage{mathtools}
\usepackage{bm}
\usepackage{bbm}
\usepackage{graphicx}
\usepackage{microtype}
\usepackage{hyperref}
\usepackage[nameinlink]{cleveref}
\usepackage{float}

\hypersetup{colorlinks=true, linkcolor=blue, citecolor=blue, urlcolor=cyan}
\newcommand{\R}{\mathbb{R}}
\renewcommand{\d}{\mathrm{d}}
\newcommand{\Hor}{\mathsf{H}}
\newcommand{\hg}{\mathsf{g}}
\newcommand{\hh}{\mathsf{h}}
\newcommand{\Tr}{\operatorname{Tr}}
\newcommand{\rank}{\operatorname{rank}}
\newcommand{\Ball}{\mathcal{B}_R}

\newtheorem{theorem}{Theorem}[section]
\newtheorem{proposition}[theorem]{Proposition}
\newtheorem{lemma}[theorem]{Lemma}
\newtheorem{corollary}[theorem]{Corollary}

\begin{document}

\title{Blind Spots of the Zwanziger Horizon Function}
\author{Daniel G. Tedesco}
\affil{\small ESEHL/PPGENT, Curitiba, PR, Brazil\\
\small \texttt{daniel.te@uninter.com}}
\date{}
\maketitle

\begin{abstract}
We examine the configurationwise relation between the first Gribov horizon, defined by loss of positivity of the Faddeev-Popov operator, and Zwanziger's horizon function, which probes the inverse operator through background-dependent sources. The analysis focuses on whether the spectral directions associated with the onset of the Gribov horizon are necessarily accessible to the sources entering the horizon function, including situations in which the critical subspace is degenerate. This question is studied for radial SU(2) hedgehog backgrounds in three and four Euclidean dimensions, where angular symmetry constrains the source sector while the radial profile controls the ordering of Faddeev-Popov thresholds. Variational estimates and finite-volume calculations are used to characterize the threshold structure for smooth radial profiles. The regular-gauge BPST background is treated separately because of domain issues associated with zero-energy behavior and the horizon source in the full-space setting. The discussion is restricted to configurationwise spectral properties and does not address the statistical weighting of these backgrounds in the Yang-Mills functional integral.
\end{abstract}

\section{Introduction}
\label{sec:question}

The Gribov ambiguity is a structural feature that distinguishes non-Abelian gauge fixing from its Abelian counterpart. In the Landau gauge, the local obstruction is encoded in the Faddeev-Popov operator,
\begin{equation}
\mathcal M^{ab}[A]= -\partial_\mu D_\mu^{ab}[A].
\label{eq:intro-fp}
\end{equation}
Zero modes of $\mathcal M$ signal the failure of the gauge condition to isolate a unique representative on a gauge orbit. Gribov therefore proposed restricting the functional measure to
\begin{equation}
\Omega=\{A\mid \partial\!\cdot\!A=0,\ \mathcal M[A]\geq0\},
\label{eq:intro-omega}
\end{equation}
whose boundary is the first Gribov horizon~\cite{Gribov1978}. Singer subsequently showed that the ambiguity is topological and cannot be removed by an alternative continuous local gauge choice~\cite{Singer1978}. The geometry associated with this restriction, including the fundamental
modular region, orbit bifurcations, and the relation between local and
absolute minima, has been studied in detail
~\cite{SemenovTyanShanskiiFranke1982,DellAntonioZwanziger1991,
VanBaal1992,Cucchieri1998,VandersickelZwanziger2012}.
Complementary orbit-space formulations characterize the geometry directly
through the metric on gauge-equivalence classes and, in explicit families,
through curves whose intersections with the Gribov horizon can be determined
analytically~\cite{Orland1996,OrlandSemenoff2000}.
Explicit topologically trivial copies inside the first Gribov region,
including spherically symmetric constructions, were subsequently revisited
in~\cite{LandimLemes2014}.

Zwanziger implemented the restriction through a nonlocal horizon functional and its local renormalizable formulation~\cite{Zwanziger1989,ZwanzigerLocal1989,Zwanziger1993}. Whenever the inverse Faddeev-Popov operator is defined on the relevant source vectors, the nonlocal quantity is
\begin{equation}
\Hor(A)=g^2\!\int\!\!\int f^{abc}A^b_\mu(x)
\bigl(\mathcal M^{-1}\bigr)^{ad}(x,y)
f^{dec}A^e_\mu(y).
\label{eq:horizon-function}
\end{equation}
Related horizon-function constructions have also been developed beyond
Landau gauge, notably in the maximal Abelian gauge, where localizability,
renormalizability, and the geometry of the corresponding Gribov region
provide additional constraints on the admissible nonlocal horizon term
~\cite{CapriLemesMAG2006,CapriLemesMAG2010}.
Gribov's no-pole condition and Zwanziger's horizon condition are functionally equivalent~\cite{CapriNoPole2013}. Their incorporation into Yang-Mills dynamics leads to infrared constructions that include the refined Gribov-Zwanziger framework, with dimension-two condensates and decoupling propagators, as well as functional approaches organized by infrared boundary conditions~\cite{Dudal2008,FischerMaasPawlowski2009}. Other prescriptions average over Gribov copies in order to study their effect on Green functions~\cite{SerreauTissier2012,Sternbeck2006,SternbeckMuller2013}.

The spectral location of the first Gribov horizon and the response of a source-sandwiched inverse need not contain the same information. A vanishing eigenvalue identifies the spectral channel in which positivity is lost, but an observable involving $\mathcal M^{-1}$ is singular only to the extent that its source has a component in the corresponding critical subspace. The reorganization of the low-lying Faddeev-Popov spectrum as the first
Gribov horizon is approached has itself been studied directly in both
Landau and Coulomb gauges~\cite{Greensite2010}. The same distinction appears in lattice bounds for the ghost propagator, where the lowest nonzero Faddeev-Popov eigenvalue enters together with the overlap of its eigenvector with the external momentum source~\cite{CucchieriMendes2008,CucchieriMendes2013,CucchieriMendesCrossing2013}. Spectral proximity to the boundary and source sensitivity to the critical direction are therefore independent data.

Exact critical modes provide a natural setting in which to separate these two questions. Henyey's construction of nontrivial Landau-gauge copies from radial profiles~\cite{Henyey1979}, subsequent finite-norm $SU(2)$ zero-mode constructions~\cite{GuimaraesSorella2011,CapriZeroModes2012}, and the regular-gauge BPST instanton~\cite{Maas2006,BPST1975,tHooft1976} give explicit backgrounds at which the Faddeev-Popov spectrum can be analyzed. The existence of a zero mode by itself, however, does not determine its overlap with the background-dependent source appearing in $\Hor(A)$.

We ask whether the first critical eigenspace encountered along a prescribed family of Landau-gauge backgrounds can be orthogonal to that source. Because a crossing may be degenerate, the question is formulated at the level of the critical subspace. Building on the operator framework developed in~\cite{TedescoSpectral,TedescoBirmanSchwinger}, we introduce a source map $T_A$ and its finite-rank visibility operator $V_A=T_AT_A^*$, then classify an isolated critical multiplet through the compression $P_cV_AP_c$. For explicit $SU(2)$ hedgehog families, angular symmetry fixes the source sector while the radial profile orders the Faddeev-Popov thresholds. This makes it possible to construct smooth profiles whose first crossing is source-dark in three dimensions and on both amplitude branches in four dimensions. The regular-gauge BPST background then identifies the point at which the full-space $L^2$ source construction ceases to apply. These statements are configurationwise and concern the spectral relation between $\mathcal M[A]$ and its horizon source, without addressing the statistical weight of the configurations in the Yang-Mills functional integral.

\section{Source visibility and critical singularities}
\label{sec:visibility}

Because the horizon functional probes the inverse Faddeev-Popov operator through finitely many background-dependent source columns, the source-label space should be kept distinct from the ghost Hilbert space. Let $X$ denote either the Dirichlet ball $\Ball$ or $\R^d$, and set
\begin{equation}
\mathcal H_{\rm gh}=L^2\!\left(X,\mathfrak{su}(N)_{\mathbb C}\right).
\label{eq:ghost-hilbert-space}
\end{equation}
Fix boundary conditions and background regularity such that $\mathcal M=\mathcal M[A]$ has a positive self-adjoint realization on $\mathcal H_{\rm gh}$. The finitely many localizer labels $(\mu,d)$ span $\mathcal K\simeq\mathbb C^{d(N^2-1)}$. With respect to its orthonormal basis $e_{\mu d}$, define
\begin{equation}
\bigl(T_Ae_{\mu d}\bigr)^a(x)=g f^{a\ell d}A^\ell_\mu(x)=m^a_{\mu d}(x).
\label{eq:source-map}
\end{equation}
The source map $T_A$ takes values in the ghost Hilbert space when each column satisfies
\begin{equation}
m_{\mu d}\in\mathcal H_{\rm gh}
\qquad\text{for every }(\mu,d).
\label{eq:source-admissibility}
\end{equation}
Since $\mathcal K$ is finite dimensional, \eqref{eq:source-admissibility} makes $T_A$ bounded and Hilbert-Schmidt. The positive operator
\begin{equation}
V_A=T_AT_A^*\geq0
\label{eq:visibility-definition}
\end{equation}
will be called the \emph{visibility operator}. Its range is the source-accessible subspace of the ghost Hilbert space. Consequently, $V_A$ is finite rank and trace class, with
\begin{equation}
\rank V_A\leq\dim\mathcal K,
\qquad
\Tr_{\mathcal H_{\rm gh}}V_A
=\|T_A\|_{\mathrm{HS}}^2
=\sum_{\mu,d}\|m_{\mu d}\|_{\mathcal H_{\rm gh}}^2.
\label{eq:visibility-trace-class}
\end{equation}
In particular, $A\in L^2(X)$ is sufficient for source admissibility.

For $\epsilon>0$, define the regularized source-sandwiched inverse form
\begin{equation}
\Hor_\epsilon(A)
=\Tr_{\mathcal K}\!\left[T_A^*(\mathcal M+\epsilon)^{-1}T_A\right]
=\left\|(\mathcal M+\epsilon)^{-1/2}T_A\right\|_{\mathrm{HS}}^2.
\label{eq:regularized-horizon}
\end{equation}
For every $\epsilon>0$ the resolvent is bounded, and \eqref{eq:regularized-horizon} is therefore finite under \eqref{eq:source-admissibility}. The limit $\epsilon\downarrow0$ is controlled by the spectral measure $E_{\mathcal M}$, for which the domain of the inverse square root is
\begin{equation}
\mathcal D(\mathcal M^{-1/2})
=
\left\{
\psi\in(\ker\mathcal M)^\perp
\,\middle|\,
\int_{(0,\infty)}\frac{1}{\lambda}\,
\d\langle\psi,E_{\mathcal M}(\lambda)\psi\rangle<\infty
\right\}.
\label{eq:inverse-form-domain}
\end{equation}

\begin{proposition}[Source-sandwiched trace identity]
\label{prop:trace-identity}
Assume that $\mathcal M\geq0$ is self-adjoint and that \eqref{eq:source-admissibility} holds. The limit
\begin{equation}
\Hor(A)=\lim_{\epsilon\downarrow0}\Hor_\epsilon(A)
\label{eq:horizon-form-limit}
\end{equation}
exists in $[0,\infty]$. It is finite if and only if
\begin{equation}
\operatorname{ran}T_A\subset\mathcal D(\mathcal M^{-1/2}).
\label{eq:inverse-form-admissibility}
\end{equation}
Under this condition, $B_A=\mathcal M^{-1/2}T_A$ is Hilbert-Schmidt and
\begin{equation}
\Hor(A)
=\Tr_{\mathcal K}(B_A^*B_A)
=\Tr_{\mathcal H_{\rm gh}}(B_AB_A^*)
=\Tr_{\mathcal H_{\rm gh}}\!\left(\mathcal M^{-1/2}V_A\mathcal M^{-1/2}\right),
\label{eq:sandwiched-trace-form}
\end{equation}
where the last operator is defined by the factorization $B_AB_A^*$. If there exists $\lambda_A>0$ such that $\mathcal M\geq\lambda_A I$, then $\mathcal M^{-1}$ is bounded and
\begin{equation}
{
\Hor(A)=\Tr_{\mathcal K}\!\left(T_A^*\mathcal M^{-1}T_A\right)
=\Tr_{\mathcal H_{\rm gh}}\!\left(\mathcal M^{-1}V_A\right).
}
\label{eq:trace-form}
\end{equation}
\end{proposition}

\begin{proof}
Put $B_{A,\epsilon}=(\mathcal M+\epsilon)^{-1/2}T_A$. Boundedness of the regularized inverse and finite dimensionality of $\mathcal K$ make $B_{A,\epsilon}$ Hilbert-Schmidt. Its two natural products are
\begin{equation}
B_{A,\epsilon}^*B_{A,\epsilon}
=T_A^*(\mathcal M+\epsilon)^{-1}T_A,
\qquad
B_{A,\epsilon}B_{A,\epsilon}^*
=(\mathcal M+\epsilon)^{-1/2}V_A(\mathcal M+\epsilon)^{-1/2}
\end{equation}
Both products are trace class and have the same trace, $\|B_{A,\epsilon}\|_{\mathrm{HS}}^2$. Applying the spectral theorem to one source column gives
\begin{equation}
\left\|(\mathcal M+\epsilon)^{-1/2}m_{\mu d}\right\|^2
=
\int_{[0,\infty)}\frac{1}{\lambda+\epsilon}\,
\d\langle m_{\mu d},E_{\mathcal M}(\lambda)m_{\mu d}\rangle.
\label{eq:spectral-source-integral}
\end{equation}
As $\epsilon$ decreases, $(\lambda+\epsilon)^{-1}$ increases pointwise. Monotone convergence therefore yields \eqref{eq:horizon-form-limit}, and the limit is finite exactly when every source column belongs to $\mathcal D(\mathcal M^{-1/2})$, which is \eqref{eq:inverse-form-admissibility}. Because $\mathcal K$ is finite dimensional, $B_A=\mathcal M^{-1/2}T_A$ is then Hilbert-Schmidt, proving \eqref{eq:sandwiched-trace-form}. If $\mathcal M\geq\lambda_A I$, both inverse powers are bounded and cyclicity of the trace reduces the expression to \eqref{eq:trace-form}.
\end{proof}

On a Dirichlet ball, smooth backgrounds automatically satisfy \eqref{eq:source-admissibility}. Inside the positive region the lowest Dirichlet eigenvalue is strictly positive, so \eqref{eq:trace-form} is an ordinary trace identity. At a boundary configuration $A_c$ where zero is an isolated eigenvalue and the complementary spectrum remains separated from zero,
\begin{equation}
\mathcal D\!\left(\mathcal M[A_c]^{-1/2}\right)=(\ker\mathcal M[A_c])^\perp.
\end{equation}
The domain relation shows that the horizon form stays finite at $A_c$ exactly when the source has no component in the kernel, $P_cT_{A_c}=0$, or equivalently $P_cV_{A_c}P_c=0$. The regularized pole is
\begin{equation}
\Hor_\epsilon(A_c)=\frac{1}{\epsilon}\Tr(P_cV_{A_c})+O(1).
\label{eq:kernel-pole}
\end{equation}
This kernel criterion is specific to an isolated zero. In full space, continuous spectrum can extend to the origin, and orthogonality to $\ker\mathcal M$ must be supplemented by $\lambda^{-1}$-integrability of the source spectral measure, exactly as required by \eqref{eq:inverse-form-admissibility}.

Let $A(\tau)$, $\tau\in(0,\tau_0]$, be a family in the Landau slice with $\mathcal M(\tau)$ self-adjoint, with lowest eigenvalues $\lambda_1(\tau),\dots,\lambda_r(\tau)\to0^+$ as $\tau\to0^+$, with a uniform gap $\inf_{\tau\in(0,\tau_0]}\lambda_{r+1}(\tau)=\lambda_\star>0$, and with $A(\tau)\to A_c$ in a norm for which $V_{A(\tau)}\to V_{A_c}$ in trace norm. Assume in addition that the convergence of the operators is in norm-resolvent sense,
\begin{equation}
\bigl\|(\mathcal M(\tau)+I)^{-1}-(\mathcal M(0)+I)^{-1}\bigr\|\longrightarrow0
\qquad(\tau\downarrow0),
\label{eq:norm-resolvent}
\end{equation}
which guarantees norm convergence of the spectral projectors associated with the isolated critical cluster. Let $P_c$ be the spectral projector of $\mathcal M(0)=\mathcal M[A_c]$ at $0$, of rank $r$, and write
\begin{equation}
V_c\equiv P_cV_{A_c}P_c .
\label{eq:critical-block}
\end{equation}

\begin{proposition}[Visibility criterion and classification]
\label{prop:visibility}
Write $P_c(\tau)$ for the spectral projector of $\mathcal M(\tau)$ onto $\lambda_1(\tau),\dots,\lambda_r(\tau)$, so that $P_c(\tau)\to P_c$ in norm. Under the hypotheses above:
\begin{enumerate}
\item[(a)] \emph{(divergence)} If $V_c\neq0$, equivalently $T_{A_c}^{*}P_c\neq0$, then
\begin{equation}
\Hor\bigl(A(\tau)\bigr)\longrightarrow\infty ,
\label{eq:criterion}
\end{equation}
and for a simple crossing, $r=1$ with normalized critical vector $\psi_c$,
\begin{equation}
\lambda_1(\tau)\,\Hor\bigl(A(\tau)\bigr)\longrightarrow
\bigl\langle\psi_c,V_{A_c}\psi_c\bigr\rangle
=\sum_{\mu,d}\bigl|\langle\psi_c,m_{\mu d}\rangle\bigr|^2>0 .
\label{eq:residue}
\end{equation}
If in addition the quadratic forms of $\mathcal M(\tau)$ share a common form domain and are differentiable at $\tau=0$ in the norm of bounded forms on it, with $\Gamma=P_c\dot{\mathcal M}(0)P_c$ positive definite on $\operatorname{ran}P_c$, then $\tau\Hor\to\Tr(\Gamma^{-1}V_c)$.
\item[(b)] \emph{(boundedness)} If $P_c(\tau)\,V_{A(\tau)}\,P_c(\tau)=0$ for every $\tau\in(0,\tau_0]$, we call the approach \emph{uniformly dark}. The singular part of $\Hor$ then vanishes identically and
\begin{equation}
\Hor\bigl(A(\tau)\bigr)\;\leq\;\frac{\Tr V_{A(\tau)}}{\lambda_\star}
\qquad\text{for all small }\tau .
\label{eq:bounded}
\end{equation}
\end{enumerate}
Crossings are classified by
\begin{equation}
\rank V_c=0:\ \text{dark};
\qquad
0<\rank V_c<r:\ \text{partially visible};
\qquad
\rank V_c=r:\ \text{fully visible} .
\label{eq:classification}
\end{equation}
\end{proposition}

\begin{proof}
Separate the critical cluster from the spectrum above the uniform gap:
\begin{equation}
\Hor=\underbrace{\sum_{j\leq r}\frac{\langle\psi_j(\tau),V_{A(\tau)}\psi_j(\tau)\rangle}{\lambda_j(\tau)}}_{\textstyle \Hor_{\rm sing}}
\;+\;\underbrace{\sum_{j>r}\frac{\langle\psi_j(\tau),V_{A(\tau)}\psi_j(\tau)\rangle}{\lambda_j(\tau)}}_{\textstyle \Hor_{\rm reg}},
\label{eq:spectral-horizon}
\end{equation}
The gap gives $0\leq\Hor_{\rm reg}\leq\Tr V_{A(\tau)}/\lambda_\star$, uniformly for small $\tau$, while the critical contribution can be written as $\Hor_{\rm sing}=\Tr\bigl(\mathcal M(\tau)^{-1}P_c(\tau)V_{A(\tau)}P_c(\tau)\bigr)$.

(a) Norm-resolvent convergence~\eqref{eq:norm-resolvent}, together with the uniform gap that isolates the critical cluster, makes the spectral projector converge in norm, $P_c(\tau)\to P_c$~\cite{Kato1980}. Trace-norm convergence of the visibility operator then gives
\begin{equation}
\sum_{j\leq r}\langle\psi_j(\tau),V_{A(\tau)}\psi_j(\tau)\rangle
=\Tr\bigl(P_c(\tau)V_{A(\tau)}\bigr)
\longrightarrow\Tr V_c.
\end{equation}
If $V_c\neq0$, the limiting trace is positive. Hence at least one critical numerator remains bounded away from zero while all critical eigenvalues tend to zero, which forces $\Hor_{\rm sing}\to\infty$. For a simple crossing, multiplying \eqref{eq:spectral-horizon} by $\lambda_1$ suppresses the bounded regular term and yields the residue \eqref{eq:residue}. In the $C^1$ case the critical block of the resolvent is controlled by a Schur complement~\cite{Kato1980}. The reduced operator on $\operatorname{ran}P_c(\tau)$ is $\tau\Gamma+o(\tau)$ in norm, so that
\begin{equation}
\tau\,P_c(\tau)\,\mathcal M(\tau)^{-1}P_c(\tau)\;\longrightarrow\;P_c\Gamma^{-1}P_c ,
\label{eq:schur-limit}
\end{equation}
and since $\Hor_{\rm reg}$ is bounded, $\tau\Hor\to\Tr(\Gamma^{-1}V_c)=\Tr(\Gamma^{-1/2}V_c\Gamma^{-1/2})\geq0$, with equality iff $V_c=0$. Under $C^1$ regularity the remainder in~\eqref{eq:schur-limit} is only $o(\tau^{-1})$, since an $O(1)$ remainder would require second-order control that is neither needed nor assumed, and nothing below uses this limit. Residues are stated in the eigenvalue normalization so that their definition does not depend on the parametrization of the path. For a transversal crossing, $\lambda_1\sim\tau\dot\lambda_1(0)$, and normalization by $\lambda_1$ differs from normalization by $\tau$ only through the nonzero crossing slope. For a tangential approach, $\dot\lambda_1(0)=0$ and $\lambda_1\sim\tfrac12\ddot\lambda_1(0)\tau^2$, so the divergence expressed in $\tau$ changes while the eigenvalue-normalized residue in \eqref{eq:residue} remains unchanged.

(b) If $P_c(\tau)V_{A(\tau)}P_c(\tau)=0$ then $\Hor_{\rm sing}=0$ identically and only $\Hor_{\rm reg}$ survives, which is~\eqref{eq:bounded}.

The factorization $V_c=(T_{A_c}^{*}P_c)^{*}(T_{A_c}^{*}P_c)$ gives $V_c=0$ if and only if $T_{A_c}^{*}P_c=0$, and it also gives $\rank V_c=\rank(T_{A_c}^{*}P_c)$. This proves the rank classification \eqref{eq:classification}. If $0<\rank V_c<\rank P_c$, the critical subspace contains both source-visible and source-dark directions, and the singular contribution is supported on the image of $T_{A_c}^{*}P_c$. Since $\rank V_c\leq\dim\mathcal K$, a critical degeneracy larger than the source-label dimension cannot be fully visible.
\end{proof}

The divergent and bounded cases require different information. The endpoint condition $V_c=0$ does not by itself control the horizon function along an approaching family because the source projection may vanish more slowly than the critical eigenvalue. For example, take $\mathcal M(\tau)=\operatorname{diag}(\tau^2,1,1,\dots)$ and $V_{A(\tau)}=|v(\tau)\rangle\langle v(\tau)|$ with $v(\tau)=\sqrt\tau\,e_1+e_2$. Then $\lambda_1=\tau^2\to0$, the complementary gap is one, and $V_{A(\tau)}\to|e_2\rangle\langle e_2|$ in trace norm, so $V_c=0$. Nevertheless, $\Hor_{\rm sing}=\tau/\tau^2\to\infty$. The hedgehog rays studied below require no auxiliary rate estimate because $V_A$ annihilates the dark symmetry sector at every amplitude. Norm-resolvent convergence controls the critical projectors in the divergent case, while exact annihilation of the critical block along the path removes the singular contribution in the dark case.

For numerical applications it is useful to associate with a normalized low Faddeev-Popov mode $\psi_1$ the source visibility
\begin{equation}
\nu_1(A)=\langle\psi_1,V_A\psi_1\rangle
=\sum_{\mu,d}|\langle\psi_1,m_{\mu d}\rangle|^2.
\label{eq:visibility-diagnostic}
\end{equation}
Near a simple visible crossing, $\lambda_1\Hor\to\nu_1(A_c)$, whereas $\lambda_1\to0$ by itself does not determine the behavior of the horizon functional~\cite{Cucchieri1998,CucchieriMendes2008,CucchieriMendes2013}. In the symmetric families considered below, angular selection can set $\nu_1$ identically to zero in a critical sector.

\section{Three-dimensional hedgehog}
\label{sec:hedgehog}

Radial and spherically symmetric $SU(2)$ backgrounds provide explicit
Faddeev-Popov zero modes, Gribov copies, and horizon bifurcations
~\cite{Henyey1979,VanBaal1992,LandimLemes2014,
GuimaraesSorella2011,CapriZeroModes2012}.
Their rotational symmetry separates source support from radial threshold
ordering. On $\Ball\subset\R^3$, consider
\begin{equation}
A_i^c(x)=\varepsilon_{cij}x_j\,h(r),
\qquad r=|\bm x| ,
\label{eq:hedgehog}
\end{equation}
with $h$ a differentiable radial profile. Antisymmetry of $\varepsilon_{cij}$ makes the field transverse for every differentiable profile, since
\begin{equation}
\partial_iA_i^c=\varepsilon_{cij}\delta_{ij}h+\varepsilon_{cij}x_jx_i\frac{h'}{r}=0 .
\label{eq:transverse}
\end{equation}
Varying only the amplitude of $h$ therefore keeps the whole ray inside the Landau-gauge slice.

With $f^{abc}=\varepsilon^{abc}$ and $D^{ab}_\mu=\delta^{ab}\partial_\mu-g\varepsilon^{abc}A^c_\mu$, transversality removes the derivative of the background and
\begin{equation}
\mathcal M^{ab}=-\delta^{ab}\partial^2+g\,\varepsilon^{abc}A^c_i\partial_i
=-\delta^{ab}\partial^2+g\,h(r)\,\varepsilon^{abc}\varepsilon_{cij}x_j\partial_i .
\label{eq:M-raw}
\end{equation}
Contracting the two Levi-Civita tensors with $\varepsilon^{abc}\varepsilon_{ijc}=\delta_{ai}\delta_{bj}-\delta_{aj}\delta_{bi}$ rewrites the interaction term as $g\,h(r)(x_b\partial_a-x_a\partial_b)$. With the orbital and color-spin generators
\begin{equation}
L_i=-i\varepsilon_{ijk}x_j\partial_k,
\qquad
(S_i)_{ab}=-i\varepsilon_{iab},
\label{eq:generators}
\end{equation}
one has $\varepsilon_{iab}L_i=-i\varepsilon_{iab}\varepsilon_{ijk}x_j\partial_k=-i(x_a\partial_b-x_b\partial_a)$, hence $x_b\partial_a-x_a\partial_b=-i\varepsilon_{iab}L_i=(S_iL_i)_{ab}$. Writing $\hh(r)\equiv g\,h(r)$, which is the only combination in which the coupling and the profile enter,
\begin{equation}
\mathcal M=-\partial^2+\hh(r)\,\bm S\!\cdot\!\bm L .
\label{eq:normal-form}
\end{equation}
The normal form \eqref{eq:normal-form} commutes with $\bm J=\bm L+\bm S$, with $\bm L^2$, and with the radial operator. Each Faddeev-Popov channel can therefore be labelled by fixed $(J,L)$. For the adjoint color representation, $S=1$, and the operator $\bm S\!\cdot\!\bm L$ has the channel eigenvalue
\begin{equation}
\bm S\!\cdot\!\bm L\big|_{(J,L)}=c_{JL}=\tfrac12\bigl[J(J+1)-L(L+1)-2\bigr],
\qquad
c_{L-1,L}=-(L+1),\quad c_{L,L}=-1,\quad c_{L+1,L}=L ,
\label{eq:cJL}
\end{equation}
with multiplicity $2J+1$ and parity $(-1)^L$. For $L=0$ only $J=1$ occurs and $c_{10}=0$, so this channel is free and cannot generate a threshold. In terms of $u(r)=r\,f(r)$ for $\psi^a=f(r)Y^a_{JLM}(\hat x)$, the reduced radial problem is
\begin{equation}
-u''+\frac{L(L+1)}{r^2}u+c_{JL}\,\hh(r)\,u=\lambda u ,
\qquad u(0)=u(R)=0 .
\label{eq:radial}
\end{equation}
Since $\bm L^2$ is conserved, the critical subspaces can be organized according to a definite orbital angular momentum, which allows the selection analysis to be carried out directly from the normal form~\eqref{eq:normal-form} and the reduced radial equation~\eqref{eq:radial}, without invoking the complete threshold structure of this family worked out in~\cite{TedescoSpectral}.

\subsection{Source support and angular weights}

The source support can be read directly from \eqref{eq:source-map}. With $\varepsilon^{a\ell c}\varepsilon_{\ell ij}=-(\delta_{ai}\delta_{cj}-\delta_{aj}\delta_{ci})$,
\begin{equation}
\bigl(T_Ae_{ic}\bigr)^a(x)=g\,\varepsilon^{a\ell c}\varepsilon_{\ell ij}x_j\,h(r)
=-\hh(r)\bigl(\delta_{ai}x_c-\delta_{ci}x_a\bigr)
=-r\,\hh(r)\,M[i,c]^a{}_k\,\hat x_k ,
\label{eq:hedgehog-source}
\end{equation}
where $M[i,c]$ is the nine-member family of $3\times3$ matrices
\begin{equation}
M[i,c]^a{}_k=\delta_{ai}\delta_{ck}-\delta_{ci}\delta_{ak} .
\label{eq:M-matrices}
\end{equation}
Each source column is a radial function times a Cartesian component of $\hat x$, so $\operatorname{ran}T_A$ lies entirely in the orbital sector $L=1$ with radial profile $r\,\hh(r)$. Parity would only exclude even $L$. Conservation of $\bm L^2$ separates $L=1$ from every higher orbital sector, including the odd ones.

\begin{lemma}[Visibility kernel and its angular spectrum]
\label{lem:kernel}
On the hedgehog~\eqref{eq:hedgehog} the visibility operator has the closed-form kernel
\begin{equation}
V_A^{ab}(x,y)=\hh(r)\hh(r')\bigl[\delta^{ab}\,(x\!\cdot\!y)+x^ay^b\bigr],
\qquad r=|x|,\ r'=|y| ,
\label{eq:visibility-kernel}
\end{equation}
so $V_A$ is supported on $L=1$ and, for any profile with $\hh\not\equiv0$, has rank $9$. Under simultaneous rotation of the color and orbital indices the nine matrices~\eqref{eq:M-matrices} decompose into trace, antisymmetric and symmetric-traceless parts, carrying $J=0,1,2$ with total weights
\begin{equation}
\sum_{i,c}\bigl\|P_J\,M[i,c]\bigr\|_F^2=w_J,
\qquad w_0=4,\quad w_1=3,\quad w_2=5 ,
\label{eq:weights}
\end{equation}
all nonzero, hence per state $\kappa_J\equiv w_J/(2J+1)=4,1,1$. Equivalently, $\mathcal K_{\rm ang}\equiv\sum_{i,c}|M[i,c]\rangle\langle M[i,c]|$ acts on $3\times3$ matrices as $4$ on the trace part and as $1$ on the eight traceless directions.
\end{lemma}

\begin{proof}
Using \eqref{eq:hedgehog-source} and summing over the source labels gives
\begin{equation}
\begin{aligned}
V_A^{ab}(x,y)
&=
\sum_{i,c}
\bigl(T_Ae_{ic}\bigr)^a(x)
\bigl(T_Ae_{ic}\bigr)^b(y) \\
&=
\hh(r)\hh(r')
\bigl[
\delta^{ab}(x\!\cdot\! y)+x^a y^b
\bigr].
\end{aligned}
\label{eq:kernel-contraction}
\end{equation}
Thus $\operatorname{ran}V_A$ is contained in the space spanned by functions of the form
$\hh(r)x_k$ with a free color index. Since each $x_k=r\hat x_k$ carries orbital angular momentum $L=1$, the visibility operator has support only in the $L=1$ sector.

To determine its angular spectrum, identify the degree-one color-angular tensors with the space of $3\times3$ matrices equipped with the Frobenius inner product. The matrices $M[i,c]$ defined in \eqref{eq:M-matrices} determine the angular operator
\begin{equation}
\mathcal K_{\rm ang}
=
\sum_{i,c}
|M[i,c]\rangle\langle M[i,c]| .
\end{equation}
For an arbitrary matrix $X$, contraction of the Kronecker symbols gives directly
\begin{equation}
\bigl(\mathcal K_{\rm ang}X\bigr)^a{}_k
=
X^a{}_k+(\Tr X)\delta^a{}_k,
\end{equation}
or, equivalently,
\begin{equation}
\mathcal K_{\rm ang}X
=
X+(\Tr X)\mathbbm 1 .
\label{eq:Kang-action}
\end{equation}

Under the diagonal $SO(3)$ action, the matrix space decomposes as
\begin{equation}
\mathrm{Mat}_3
=
\mathbb R\,\mathbbm 1
\oplus
\mathfrak{so}(3)
\oplus
\mathrm{Sym}_0(3),
\end{equation}
corresponding respectively to the $J=0,1,2$ sectors, of dimensions $1$, $3$, and $5$. Equation~\eqref{eq:Kang-action} shows that $\mathcal K_{\rm ang}$ acts with eigenvalue $4$ on the trace component and with eigenvalue $1$ on both traceless components. It is therefore nondegenerate on the full nine-dimensional angular space, so $\rank V_A=9$ whenever $\hh\not\equiv0$. The total weights are consequently
\begin{equation}
w_0=4,\qquad
w_1=3,\qquad
w_2=5,
\end{equation}
and division by the corresponding multiplet dimensions gives
\begin{equation}
\kappa_0=4,\qquad
\kappa_1=\kappa_2=1.
\end{equation}
This proves \eqref{eq:weights}.
\end{proof}

\subsection{Exact orbital selection and the Henyey control profile}

\begin{theorem}[Orbital selection rule]
\label{thm:selection}
For the hedgehog~\eqref{eq:hedgehog} with any differentiable radial profile,
\begin{equation}
V_{(J,L)}
\equiv
P_{(J,L)}V_AP_{(J,L)}
=0,
\qquad L\neq1,
\label{eq:dark-rule}
\end{equation}
so every threshold outside the $L=1$ sector is a dark crossing. In the three $L=1$ channels the crossings are fully visible,
\[
\rank V_{(J,1)}=2J+1,
\]
and for a critical state with reduced radial function $u$,
\begin{equation}
\bigl\langle\psi_{(J,1)},V_A\psi_{(J,1)}\bigr\rangle
=
\frac{4\pi}{3}\,
\kappa_J\,
\frac{|\langle\rho,u\rangle|^2}{\|u\|^2},
\qquad
\rho(r)=r^2\hh(r),
\qquad
\kappa_0=4,\quad
\kappa_1=\kappa_2=1.
\label{eq:residue-hedgehog}
\end{equation}
The full horizon function therefore reduces to
\begin{equation}
\Hor(A)
=
\frac{4\pi}{3}
\sum_{J=0}^{2}
w_J\,
\bigl\langle\rho,h_{J1}^{-1}\rho\bigr\rangle,
\qquad
h_{J1}
=
-\partial_r^2+\frac{2}{r^2}+c_{J1}\hh(r),
\label{eq:horizon-reduced}
\end{equation}
where
\[
c_{01}=-2,
\qquad
c_{11}=-1,
\qquad
c_{21}=+1,
\]
and $w_J$ are the angular weights in~\eqref{eq:weights}.
\end{theorem}

\begin{proof}
By \cref{lem:kernel}, the range of $T_A$ lies entirely in the $L=1$ eigenspace of $\bm L^2$. Since $\bm L^2$ commutes with $\mathcal M$, its eigenspaces are invariant under the spectral projectors of $\mathcal M$. Hence every critical subspace with $L\neq1$ is orthogonal to $\operatorname{ran}T_A$, and therefore
\[
P_{(J,L)}T_A=0,
\qquad L\neq1.
\]
Equation~\eqref{eq:dark-rule} follows immediately from $V_A=T_AT_A^\ast$. The darkness of these channels is thus an exact subspace property, rather than a cancellation involving individual critical states.

For $L=1$, write
\[
\psi^a(x)=\frac{u(r)}{r}Y^a(\hat x),
\]
with $Y$ in a fixed diagonal-$SO(3)$ multiplet $J$. The angular decomposition obtained in \cref{lem:kernel} gives
\begin{equation}
\bigl\langle\psi,V_A\psi\bigr\rangle
=
\frac{4\pi}{3}\,
\kappa_J\,
\frac{|\langle\rho,u\rangle|^2}{\|u\|^2},
\end{equation}
because $\mathcal K_{\rm ang}$ acts with eigenvalue $\kappa_J$ on that multiplet. This establishes~\eqref{eq:residue-hedgehog} once the radial overlap is shown to be nonzero.

At a zero-energy crossing, the reduced radial equation implies, after multiplication by $r^2$ and integration by parts,
\begin{equation}
c_{JL}\langle\rho,u\rangle
=
R^2u'(R)
+
\bigl[2-L(L+1)\bigr]
\int_0^R u(r)\,\d r .
\end{equation}
The second term vanishes precisely for $L=1$. On a finite interval with Dirichlet boundary conditions,
\[
c_{J1}\langle\rho,u\rangle=R^2u'(R),
\]
which cannot vanish for a nontrivial solution, since $u(R)=u'(R)=0$ would imply $u\equiv0$. On $\R^3$, a normalizable zero-energy $L=1$ state has $u(r)\sim\beta/r$, and the corresponding identity becomes
\[
c_{J1}\langle\rho,u\rangle=-3\beta\neq0.
\]
Thus the restriction of $V_A$ to every $L=1$ critical multiplet is positive definite and
\[
\rank V_{(J,1)}=2J+1.
\]

Finally, each source column has the same reduced radial profile $\rho(r)$ and only $L=1$ angular support. Resolving the source space into its $J=0,1,2$ components therefore gives~\eqref{eq:horizon-reduced}, with the multiplicities encoded by $w_J$. Near a simple $(J,1)$ threshold,
\[
h_{J1}^{-1}
\simeq
\frac{|u\rangle\langle u|}
{\lambda\,\|u\|^2},
\]
so the singular part of~\eqref{eq:horizon-reduced} reproduces~\eqref{eq:residue-hedgehog} after summing over the $2J+1$ states of the multiplet.
\end{proof}

Henyey's profile provides a useful closed-form check of the selection rule while also showing that darkness does not by itself determine which channel crosses first. Set
\[
\hh(r)=\hg\,\varphi(r),
\qquad
\varphi(r)=\frac{9r}{(r^3+1)^2}.
\]
For
\[
u_L(r)
=
r^{L+1}(r^3+1)^{-(2L+1)/3},
\]
direct substitution into the free radial operator gives
\begin{equation}
\left(
-\partial_r^2+\frac{L(L+1)}{r^2}
\right)u_L
=
\frac{2(2L+1)(L+2)}{9}\,
\varphi(r)u_L(r).
\end{equation}
The zero-energy condition therefore fixes the threshold amplitude as
\begin{equation}
\hg_\ast(J,L)
=
-\frac{2(2L+1)(L+2)}
{9\,c_{JL}}.
\label{eq:henyey-thresholds}
\end{equation}
The lowest member reproduces Henyey's result~\cite{Henyey1979}. Since $u_L\sim r^{-L}$ at large $r$, all states with $L\geq1$ are square integrable. Combining~\eqref{eq:henyey-thresholds} with \cref{thm:selection} gives the visibility pattern summarized in \cref{tab:henyey}.

\begin{table}[H]
\centering
\renewcommand{\arraystretch}{1.25}
\begin{tabular}{cccccl}
\hline
$J^P$ & $L$ & $c_{JL}$ & $\hg_\ast$ & $\rank V_c$ & role along the ray\\
\hline
$0^-$ & $1$ & $-2$ & $1$     & $1$ & first positive crossing, fully visible\\
$1^+$ & $2$ & $-3$ & $40/27$ & $0$ & second positive threshold, dark\\
$2^-$ & $3$ & $-4$ & $35/18$ & $0$ & third positive threshold, dark\\
$1^-$ & $1$ & $-1$ & $2$     & $3$ & fourth positive threshold, fully visible\\
$2^+$ & $2$ & $-1$ & $40/9$  & $0$ & higher, dark\\
$2^-$ & $1$ & $+1$ & $-2$    & $5$ & first negative crossing, fully visible\\
\hline
\end{tabular}
\caption{Closed-form Henyey thresholds and their visibility under the orbital selection rule. The stability interval is $\hg\in(-2,1)$, whose two endpoints are fully visible $L=1$ crossings. The first dark positive threshold occurs in the $(J,L)=(1,2)$ channel.}
\label{tab:henyey}
\end{table}

The table separates threshold ordering from visibility. In particular, the first positive crossing is visible, whereas the next one is dark. It also shows that total angular momentum and parity do not determine visibility by themselves: the two $2^-$ channels have different visibility because they belong to different orbital sectors.

The first positive crossing also supplies a normalization check for the residue. At $\hg=1$, the Henyey zero mode is
\[
u_1(r)=\frac{r^2}{r^3+1},
\]
for which the radial overlap and norm are
\begin{equation}
\langle\rho,u_1\rangle=\frac{3}{2},
\qquad
\|u_1\|^2=\frac{4\pi}{9\sqrt3}.
\label{eq:radial-overlap}
\end{equation}
With $\kappa_0=4$, equation~\eqref{eq:residue-hedgehog} then gives $\lambda_1\Hor \longrightarrow 27\sqrt3 \quad(\hg\to1^-)$,
in units where the Henyey scale is unity. The finite-volume calculation converges to the same value when the threshold is determined before taking $R\to\infty$, with the expected $O(R^{-1})$ correction from truncating the critical $r^{-1}$ tail.

\section{Three-dimensional threshold ordering and dark first crossings}
\label{sec:dark}

Henyey's profile contains dark higher-$L$ channels, but both endpoints of its stability interval occur in the visible $L=1$ sector. Obtaining a dark first crossing therefore requires a different ordering of the radial thresholds. This ordering can be changed by varying the radial profile without affecting the source selection rule, since \eqref{eq:hedgehog-source} remains purely $L=1$ along the whole amplitude ray.

Consider a ray $\hh=\hg\varphi$ with $\varphi\geq0$. Attraction requires $c_{JL}\hg<0$, and \eqref{eq:cJL} shows that the strongest attractive channels at fixed $L$ are $c_{L-1,L}=-(L+1)$ for $\hg>0$ and $c_{L+1,L}=L$ for $\hg<0$. The two endpoints are therefore
\begin{equation}
\hg_\ast^+=\min_{L\geq1}\frac{\mu_L[\varphi]}{L+1},
\qquad
\bigl|\hg_\ast^-\bigr|=\min_{L\geq1}\frac{\mu_L[\varphi]}{L},
\label{eq:endpoints}
\end{equation}
where $\mu_L[\varphi]$ is the attractive coupling at which the radial quadratic form first loses positivity. On the half-line the essential spectrum already reaches zero, so the endpoint is characterized by this loss of positivity rather than by a downward motion of the spectral bottom. Its visibility is then fixed by the angular momentum $L$ that minimizes \eqref{eq:endpoints}.

To treat the delta shell and smooth narrow profiles in the same framework, let $\nu$ be a finite positive Borel measure supported in $[a,b]\subset(0,\infty)$. For $L\geq1$, define
\begin{equation}
Q_L[u]
=
\int_0^\infty |u'|^2\,\d r
+L(L+1)\int_0^\infty\frac{u^2}{r^2}\,\d r,
\label{eq:quadratic-form}
\end{equation}
and let
\begin{equation}
\mathcal Q_L
=
\overline{C_c^\infty(0,\infty)}^{\;Q_L^{1/2}}
\label{eq:energy-space}
\end{equation}
be the corresponding homogeneous energy space. Unlike the Friedrichs form domain
\[
\mathcal F_L=\mathcal Q_L\cap L^2(0,\infty),
\]
$\mathcal Q_L$ may contain non-square-integrable functions. This distinction is useful because the threshold variational problem is naturally formulated in $\mathcal Q_L$.

Since $\nu$ has compact support away from the origin, restriction to $[a,b]$ is continuous in the homogeneous energy norm. In particular, there is a finite constant $C_L(a,b)$ such that
\begin{equation}
\sup_{r\in[a,b]}|u(r)|^2
\leq
C_L(a,b)\,Q_L[u],
\qquad u\in\mathcal Q_L,
\label{eq:eval-bound}
\end{equation}
and, on $\mathcal F_L$, for every $\epsilon>0$ one has
\begin{equation}
\int u^2\,\d\nu
\leq
\epsilon Q_L[u]+C_\epsilon\|u\|^2
\label{eq:eval-infinitesimal}
\end{equation}
with some finite $C_\epsilon$. Thus the measure perturbation is infinitesimally form-bounded with respect to the free Friedrichs form, and the KLMN theorem defines a self-adjoint, lower-bounded operator $H_{L,\mu,\nu}$ associated with
\[
Q_L[u]-\mu\int u^2\,\d\nu.
\]
The continuity in \eqref{eq:eval-bound}, together with the density of $C_c^\infty(0,\infty)$ in $\mathcal Q_L$, also implies that minimizing over $\mathcal Q_L$ gives the same stability threshold as minimizing over $\mathcal F_L$.

For compactly supported $\nu$, the essential spectrum remains $[0,\infty)$. The threshold is therefore the value at which the quadratic form ceases to be nonnegative. The next lemma makes this statement precise and shows that the threshold is attained by an $L^2$ zero mode.

\begin{lemma}[Variational characterization of the threshold]
\label{lem:variational}
For $L\geq1$, define
\begin{equation}
\mu_L[\nu]
=
\inf\left\{
\frac{Q_L[u]}{\int u^2\,\d\nu}
:\; u\in\mathcal Q_L,
\ \int u^2\,\d\nu>0
\right\}.
\label{eq:variational}
\end{equation}
Then:
\begin{enumerate}
\item[(i)] the infimum is attained by a positive function $u_L$. Outside $\operatorname{supp}\nu$, the minimizer solves the free radial equation and satisfies $u_L(r)=O(r^{-L})$ as $r\to\infty$, hence $u_L\in L^2(0,\infty)$;
\item[(ii)] $Q_L[u]-\mu\int u^2\d\nu$ is nonnegative for $\mu\leq\mu_L[\nu]$ and takes negative values for $\mu>\mu_L[\nu]$;
\item[(iii)] at $\mu=\mu_L[\nu]$, the minimizer is a zero-energy eigenfunction of $H_{L,\mu_L,\nu}$.
\end{enumerate}
\end{lemma}

\begin{proof}
Normalize a minimizing sequence by $\int u_n^2\d\nu=1$ and replace $u_n$ by $|u_n|$. Boundedness of $Q_L[u_n]$ gives weak compactness in $\mathcal Q_L$. On the compact support of $\nu$, the restriction map into $H^1(a,b)$ is bounded and the embedding $H^1(a,b)\hookrightarrow C([a,b])$ is compact. A subsequence therefore converges uniformly on $[a,b]$, so the normalization passes to the weak limit, while lower semicontinuity of $Q_L$ shows that the limit attains the infimum.

Outside $\operatorname{supp}\nu$, the Euler--Lagrange equation is free. Its radial solutions are $r^{L+1}$ and $r^{-L}$, with finite energy selecting the decaying branch at infinity. Since $L\geq1$, the resulting $r^{-L}$ tail is square integrable. Part~(ii) follows directly from the definition of the Rayleigh quotient, and stationarity at the minimizer gives the weak eigenvalue equation at $\mu=\mu_L$. Because the minimizer belongs to $L^2$, this weak solution is an actual zero-energy eigenfunction of $H_{L,\mu_L,\nu}$.
\end{proof}

When $\nu$ has a density $\varphi$, we continue to write $\mu_L[\varphi]$. The analytic thresholds above are defined on the half-line, whereas the numerical calculation uses a Dirichlet box. The finite-volume correction is controlled by the following estimate.

\begin{lemma}[Box versus half-line]
\label{lem:box-vs-line}
Let $\mu_L^{(R)}[\nu]$ be defined by \eqref{eq:variational} with the additional boundary condition $u(R)=0$, where $R>b$. Then $\mu_L^{(R)}\geq\mu_L$, the function $R\mapsto\mu_L^{(R)}$ is nonincreasing, and
\begin{equation}
0\leq
\frac{\mu_L^{(R)}[\nu]-\mu_L[\nu]}{\mu_L[\nu]}
\leq
C_{L,\nu}
\frac{(b/R)^{2L+1}}{1-(b/R)^{2L+1}},
\qquad
C_{L,\nu}
=
\frac{(2L+1)\beta^2b^{-2L-1}}{Q_L[u_L]},
\label{eq:box-correction}
\end{equation}
where $u_L(r)=\beta r^{-L}$ for $r\geq b$.
\end{lemma}

\begin{proof}
The inequality $\mu_L^{(R)}\geq\mu_L$ follows because the Dirichlet condition restricts the variational class, while monotonicity follows by enlarging the box. To estimate the difference, keep the half-line minimizer unchanged on $[0,b]$ and replace its tail by the unique free solution matching $u_L(b)$ and vanishing at $R$,
\[
v_R(r)
=
\frac{\beta}{1-(b/R)^{2L+1}}
\left(r^{-L}-R^{-2L-1}r^{L+1}\right).
\]
The interaction term is unchanged because $\nu$ is supported in $[a,b]$. Evaluating the free tail contribution and inserting this function into the box Rayleigh quotient gives \eqref{eq:box-correction}. For fixed $L$, the relative correction is therefore $O(R^{-2L-1})$.
\end{proof}

The finite-volume effect is largest in the lowest angular channels. For the representative choice $r_0=10$ and $w=0.4$, the $L=1$ data follow the expected $R^{-3}$ convergence, and at $R=125$ the relative shift is of order $5\times10^{-4}$. This is much smaller than the channel separations relevant below, which are of order $15\%$--$26\%$. The box therefore does not alter the threshold ordering in the parameter range used here, although its larger correction at small $L$ biases finite-$R$ estimates slightly toward higher angular momentum.

The opposite ordering of the positive and negative amplitude branches can already be understood from their angular factors. Increasing $L$ raises the centrifugal contribution to \eqref{eq:variational}, but for a profile concentrated near a large radius $r_0$ this cost is suppressed by $r_0^{-2}$. In the delta-shell limit, $\mu_L$ then grows only linearly with $L$. On the negative branch the attractive color-spin factor $|c_{L+1,L}|=L$ grows at the same rate, allowing $\mu_L/L$ to decrease with $L$.

\begin{lemma}[Delta-shell thresholds]
\label{lem:delta-shell}
For the unit shell $\nu=\delta_{r_0}$ and $L\geq1$,
\begin{equation}
\mu_L=\frac{2L+1}{r_0},
\qquad
u_L(r)
=
\begin{cases}
(r/r_0)^{L+1}, & r\leq r_0,\\[2pt]
(r_0/r)^L, & r\geq r_0,
\end{cases}
\label{eq:delta-shell}
\end{equation}
where $u_L$ is nodeless and
\[
\|u_L\|^2
=
r_0\left(\frac{1}{2L+3}+\frac{1}{2L-1}\right).
\]
Consequently,
\begin{equation}
\begin{aligned}
\bigl|\hg_\ast^-(L)\bigr|
&=\frac{\mu_L}{L}
=\frac1{r_0}\left(2+\frac1L\right),
&&\text{strictly decreasing},\\
\hg_\ast^+(L)
&=\frac{\mu_L}{L+1}
=\frac1{r_0}\frac{2L+1}{L+1},
&&\text{strictly increasing}.
\end{aligned}
\label{eq:branches}
\end{equation}
\end{lemma}

\begin{proof}
At fixed $u(r_0)=1$, the minimization is free on either side of the shell. Regularity at the origin selects the $r^{L+1}$ solution, while square integrability at infinity selects $r^{-L}$, giving \eqref{eq:delta-shell}. Direct evaluation of the quadratic form yields
\[
Q_L[u_L]
=
\frac{L+1}{r_0}+\frac{L}{r_0}
=
\frac{2L+1}{r_0},
\]
which proves the threshold formula. The two branches in \eqref{eq:branches} then follow from $|c_{L+1,L}|=L$ and $|c_{L-1,L}|=L+1$.
\end{proof}

The delta shell therefore orders the two amplitude branches in opposite ways. On the positive branch the first threshold occurs at $L=1$ and is fully visible. On the negative branch the thresholds decrease toward the unattained limit $2/r_0$, with every $L\geq2$ channel dark. Giving the shell a finite width turns this limiting behavior into an attainable minimum while preserving quantitative control of the thresholds.

\begin{proposition}[Two-sided finite-width bound]
\label{prop:width}
Let $\varphi_w\geq0$ be bounded, normalized by $\int\varphi_w\,\d r=1$, and supported in $[r_0-w,r_0+w]$. For every $L\geq1$ with $\mu_Lw<1$,
\begin{equation}
\frac{\mu_L}{\bigl(1+\sqrt{\mu_Lw}\bigr)^2}
\leq
\mu_L[\varphi_w]
\leq
\frac{\mu_L}{\bigl(1-\sqrt{\mu_Lw}\bigr)^2},
\qquad
\mu_L=\frac{2L+1}{r_0}.
\label{eq:two-sided}
\end{equation}
Hence $\mu_L[\varphi_w]\to\mu_L$ as $w\to0$, with an $O(\sqrt w)$ bound on the convergence rate.
\end{proposition}

\begin{proof}
The two estimates follow from the pointwise control
\[
|u(r)-u(r_0)|
\leq
\sqrt{|r-r_0|}\,\|u'\|_{L^2}.
\]
For the upper bound, use the delta-shell minimizer as a trial function. Since $u_L(r_0)=1$ and $\|u_L'\|_{L^2}^2\leq Q_L[u_L]=\mu_L$, one has $u_L(r)\geq1-\sqrt{\mu_Lw}$ on the support of $\varphi_w$, which gives the right-hand inequality in \eqref{eq:two-sided}.

For the lower bound, normalize an arbitrary trial function by $Q_L[u]=1$. The same pointwise estimate gives $|u(r)|\leq|u(r_0)|+\sqrt w$ on the support of $\varphi_w$, while the delta-shell variational principle implies $|u(r_0)|\leq\mu_L^{-1/2}$. Hence
\[
\int\varphi_wu^2\,\d r
\leq
\left(\mu_L^{-1/2}+\sqrt w\right)^2.
\]
Taking the supremum over normalized trial functions and inverting gives the left-hand inequality.
\end{proof}

\subsection{Finite-width existence theorem}

\begin{theorem}[Existence of a source-dark first crossing]
\label{thm:dark}
Let $\varphi_w$ be as in \cref{prop:width}. Then:
\begin{enumerate}
\item[(i)] for every $w>0$, the minimum of $L\mapsto \mu_L[\varphi_w]/L$ over $L\geq1$ is attained at a finite value $L_\ast(w)$;
\item[(ii)] if
\begin{equation}
\frac{w}{r_0}<5.2\times10^{-4},
\end{equation}
then $L_\ast(w)\geq2$.
\end{enumerate}
For such profiles, the negative endpoint $\hg=\hg^-_\ast$ is reached by a zero mode in a channel with $L\geq2$. Hence $V_c=0$ and the horizon function remains finite. On a Dirichlet ball this endpoint belongs to the first Gribov boundary, while on the half-line it marks the loss of positivity of the quadratic form.
\end{theorem}

\begin{proof}
A negative direction can occur only if the attractive potential overcomes the centrifugal barrier somewhere on the support of $\varphi_w$. Consequently,
\begin{equation}
\frac{\mu_L[\varphi_w]}{L}
>
\frac{L+1}{\|\varphi_w\|_\infty(r_0+w)^2},
\label{eq:tail-bound}
\end{equation}
which diverges with $L$ and proves part~(i).

For part~(ii), it is enough to exclude the only visible candidate on the negative branch, namely $L=1$. The two-sided estimate of \cref{prop:width}, applied to $L=1$ and $L=2$, gives
\begin{equation}
\frac{\mu_2[\varphi_w]}{2}
\leq
\frac{5}{2r_0\bigl(1-\sqrt{5w/r_0}\bigr)^2},
\qquad
\mu_1[\varphi_w]
\geq
\frac{3}{r_0\bigl(1+\sqrt{3w/r_0}\bigr)^2}.
\label{eq:bracket}
\end{equation}
Comparing these bounds yields the sufficient condition $w/r_0<5.2\times10^{-4}$. Thus the global minimizer is finite by part~(i) and satisfies $L_\ast\geq2$.

At the negative endpoint, the minimizing attractive channel is $(J,L)=(L_\ast+1,L_\ast)$. Since $L_\ast\geq2$, \cref{thm:selection} gives $V_c=0$, so the critical zero mode is absent from the source sector. Moreover,
\[
|\hg^-_\ast|
=
\frac{\mu_{L_\ast}[\varphi_w]}{L_\ast}
<
\mu_1[\varphi_w],
\]
and therefore all three visible $L=1$ resolvents remain below their thresholds. On a Dirichlet ball this gives the bound
\begin{equation}
\Hor
\leq
\frac{\Tr V_A}
{\lambda_{\min}\!\left(\mathcal M|_{L=1}\right)},
\label{eq:sector-bound}
\end{equation}
while \cref{prop:continuum-finite} establishes the corresponding statement on $\R^3$.
\end{proof}

The same centrifugal estimate bounds the angular range that can contain the minimizing channel.

\begin{corollary}[Finite angular range]
\label{cor:Lmax}
Let $m$ be a value of $\mu_L[\varphi_w]/L$ attained by some channel. Every minimizer satisfies
\begin{equation}
L_\ast
\leq
\|\varphi_w\|_\infty(r_0+w)^2m-1.
\label{eq:Lmax}
\end{equation}
Hence only finitely many angular channels can compete for the endpoint at fixed profile.
\end{corollary}

\begin{proof}
Equation~\eqref{eq:tail-bound} implies that a channel with
\[
L+1>\|\varphi_w\|_\infty(r_0+w)^2m
\]
has $\mu_L[\varphi_w]/L>m$ and therefore cannot minimize the endpoint.
\end{proof}

To examine widths beyond the sufficient narrow-shell regime, we use the same normalized $C^\infty$ bump profile in all geometries,
\begin{equation}
\varphi_w(r)
=
\frac{1}{wI_0}
\exp\!\left[
-\frac{1}{1-\bigl((r-r_0)/w\bigr)^2}
\right]
\mathbf 1_{\{|r-r_0|<w\}},
\qquad
I_0
=
\int_{-1}^{1}e^{-1/(1-s^2)}\,\d s
\simeq0.44 .
\label{eq:bump}
\end{equation}
Its maximum is $\|\varphi_w\|_\infty=e^{-1}/(wI_0)$. The bound in \cref{thm:dark} is intentionally conservative because it follows from a uniform two-sided estimate. Direct threshold calculations show that the dark ordering survives for substantially broader shells and that the minimizing angular channel depends on the shell geometry. The examples in \cref{tab:shell} give $L_\ast=3,5,8$. For $r_0=10$, the visible $L=1$ endpoint and the dark $L=2$ endpoint coincide near
\[
\frac{w_c}{r_0}\simeq0.62,
\]
while the ordering remains dark up to $w/r_0=0.60$. No monotonicity in $w$ is assumed.

At $w=w_c$, the equality
\[
\mu_1[\varphi_w]
=
\frac{\mu_2[\varphi_w]}{2}
\]
makes the $(J,L)=(2,1)$ and $(3,2)$ channels critical at the same amplitude. Their direct sum has dimension $r=12$. The selection rule annihilates the $L=2$ septet but leaves the $L=1$ quintet fully visible, so
\begin{equation}
\rank V_c=5,
\qquad
0<\rank V_c<12.
\end{equation}
The crossing is therefore partially visible. The horizon function diverges, but it probes only five of the twelve critical directions. This realizes the intermediate-rank case of \eqref{eq:classification}; full visibility is impossible here because $\dim\mathcal K=9<12$.

\begin{table}[H]
\centering
\small
\setlength{\tabcolsep}{5pt}
\renewcommand{\arraystretch}{1.2}
\begin{tabular}{lcccccrc}
\hline
profile
& $L{=}1$
& $L{=}2$
& $L{=}3$
& $L{=}5$
& $L{=}8$
& minimizing $L$ ($|\hg^-_\ast|$)
& crossing\\
\hline
Henyey, exact
& $2.00$
& $2.22$
& $2.59$
& $3.42$
& $4.72$
& $1$ ($2.00$)
& visible\\
shell $r_0=4$, $w/r_0=0.10$
& $0.802$
& $0.698$
& $0.679$
& $0.694$
& $0.749$
& $3$ ($0.679$)
& \textbf{dark}\\
shell $r_0=10$, $w/r_0=0.04$
& $0.308$
& $0.262$
& $0.248$
& $0.243$
& $0.246$
& $5$ ($0.243$)
& \textbf{dark}\\
shell $r_0=25$, $w/r_0=0.016$
& $0.121$
& $0.102$
& $0.096$
& $0.092$
& $0.090$
& $8$ ($0.090$)
& \textbf{dark}\\
\hline
delta-shell limit
& $3/r_0$
& $5/(2r_0)$
& $7/(3r_0)$
& $11/(5r_0)$
& $17/(8r_0)$
& $L\to\infty$
& n/a\\
\hline
\end{tabular}
\caption{Negative-branch endpoints $|\hg_\ast^-(L)|=\mu_L[\varphi]/L$. Henyey's profile is minimized in the visible $L=1$ sector, whereas the representative smooth shells are minimized at the dark values $L=3,5,8$. The minimizing channel therefore depends on the radial geometry rather than on a fixed angular sector.}
\label{tab:shell}
\end{table}

Figure~\ref{fig:dark} shows the same ordering and its evolution with the shell width.

\begin{figure}[H]
\centering
\includegraphics[width=\textwidth]{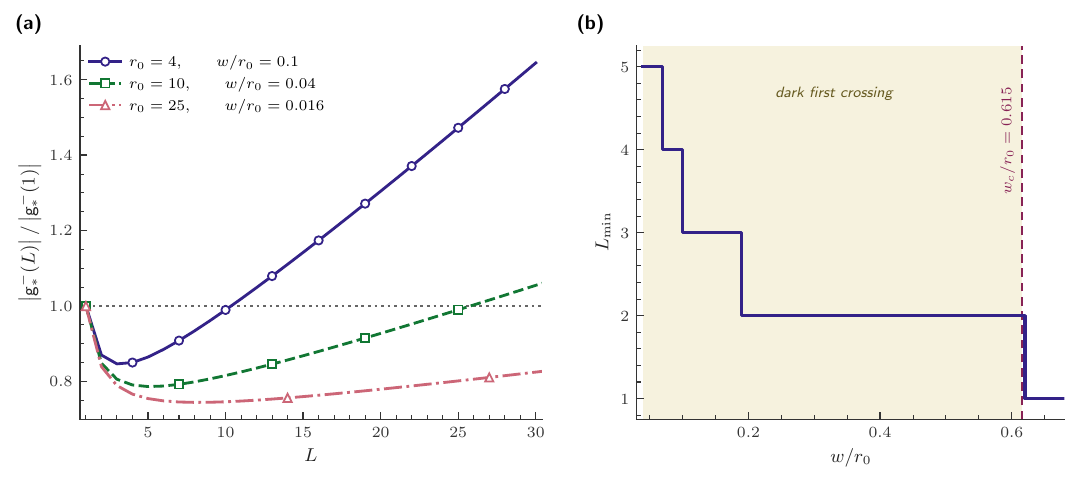}
\caption{Three-dimensional dark ordering beyond the sufficient narrow-shell regime. Panel (a) shows the negative-branch endpoints normalized to the visible $L=1$ value for representative smooth shells, with minima at $L=3,5,8$. Panel (b) follows the minimizing channel as the width varies at $r_0=10$. The visible and $L=2$ thresholds coincide near $w_c/r_0\simeq0.62$, while dark ordering persists up to $w/r_0=0.60$.}
\label{fig:dark}
\end{figure}

\begin{proposition}[Finiteness of the horizon function at a dark crossing]
\label{prop:continuum-finite}
Let $\hh=\hg\varphi$ with $\varphi\geq0$ compactly supported in $(0,\infty)$, and let $\hg=\hg^-_\ast$ be the negative endpoint of the stability interval. Assume that a minimizing channel satisfies $L_\ast\geq2$ and that the visible sector remains strictly below threshold,
\begin{equation}
|\hg^-_\ast|
=
\frac{\mu_{L_\ast}[\varphi]}{L_\ast}
<
\mu_1[\varphi].
\label{eq:strict-visible-gap-3d}
\end{equation}
Then, on $\R^3$, with
\begin{equation}
\langle\rho,h^{-1}\rho\rangle
:=
\lim_{\epsilon\downarrow0}
\langle\rho,(h+\epsilon)^{-1}\rho\rangle,
\label{eq:threshold-inverse}
\end{equation}
the horizon function remains finite,
\begin{equation}
\Hor(A)
=
\frac{4\pi}{3}
\left[
4\,\bigl\langle\rho,h_{01}^{-1}\rho\bigr\rangle
+
3\,\bigl\langle\rho,h_{11}^{-1}\rho\bigr\rangle
+
5\,\bigl\langle\rho,h_{21}^{-1}\rho\bigr\rangle
\right]
<\infty,
\label{eq:continuum-finite}
\end{equation}
even though $\mathcal M$ has an $L^2$ zero mode and $\inf\sigma(\mathcal M)=0$.
\end{proposition}

\begin{proof}
Equation~\eqref{eq:horizon-reduced} already restricts the horizon function to the three visible $L=1$ channels. It is therefore enough to control their inverse quadratic forms at the endpoint. For each channel,
\begin{equation}
\bigl\langle\rho,h_{J1}^{-1}\rho\bigr\rangle
=
\sup_v
\left\{
2\langle\rho,v\rangle-Q_{J1}[v]
\right\},
\qquad
Q_{J1}[v]
=
Q_1[v]+c_{J1}\!\int\hh\,v^2 .
\label{eq:legendre}
\end{equation}

On the negative branch, the $J=0$ and $J=1$ potentials are repulsive, so $Q_{J1}\geq Q_1$. The $J=2$ channel is attractive, but \cref{lem:variational} and the strict-gap assumption give
\begin{equation}
Q_{21}[v]
\geq
\left(
1-\frac{|\hg^-_\ast|}{\mu_1[\varphi]}
\right)
Q_1[v]
\equiv
\theta Q_1[v],
\qquad
\theta>0.
\label{eq:theta-positive}
\end{equation}
Thus all three visible forms satisfy a coercive estimate $Q_{J1}\geq\theta_JQ_1$ with $\theta_J>0$.

Because $\rho$ has compact support, the evaluation bound \eqref{eq:eval-bound} also makes the source functional continuous in the $Q_1$ norm. Hence there is a finite constant $C_\rho$ such that
\begin{equation}
|\langle\rho,v\rangle|^2
\leq
C_\rho\,Q_1[v].
\label{eq:sobolev-bound}
\end{equation}
Substitution into \eqref{eq:legendre} gives
\[
\bigl\langle\rho,h_{J1}^{-1}\rho\bigr\rangle
\leq
\sup_{q\geq0}
\left(
2\sqrt{C_\rho q}-\theta_J q
\right)
=
\frac{C_\rho}{\theta_J},
\]
which proves \eqref{eq:continuum-finite}.

The zero mode itself lies in the dark channel $(L_\ast+1,L_\ast)$ and therefore does not contribute to the source-sector trace. At the same time, compact support of the perturbation leaves the essential spectrum at $[0,\infty)$, so $\inf\sigma(\mathcal M)=0$. The finiteness of $\Hor(A)$ follows from the exact source-sector reduction together with the strict separation of the visible $L=1$ sector from its own threshold.
\end{proof}

\section{Four-dimensional 't~Hooft hedgehog}
\label{sec:fourd}

The four-dimensional construction is based on an $SO(4)$-covariant hedgehog and therefore has its own angular selection rule, rather than a static embedding of the three-dimensional ansatz. Let $\eta_{a\mu\nu}$ denote the self-dual 't~Hooft symbols, with $\eta_{a\mu\nu}=\varepsilon_{a\mu\nu}$ for $\mu,\nu\leq3$, $\eta_{a\mu4}=\delta_{a\mu}$, and $\eta_{a4\nu}=-\delta_{a\nu}$. We consider
\begin{equation}
A^a_\mu(x)=\eta_{a\mu\nu}x_\nu h(r),
\qquad
r=|x|.
\label{eq:thooft-hedgehog}
\end{equation}
Antisymmetry of $\eta_{a\mu\nu}$ makes the field transverse for every differentiable radial profile. As before, we write $\hh=gh$. The regular-gauge BPST instanton corresponds to $\hh(r)=2/(r^2+\rho^2)$~\cite{BPST1975,tHooft1976}; for the threshold analysis we use the amplitude family
\begin{equation}
\hh_{\hg}(r)=\frac{2\hg}{r^2+\rho^2},
\label{eq:bpst-amplitude-family}
\end{equation}
so that the physical BPST configuration is recovered at $\hg=1$.

\subsection{Angular normal form and source support}

Introduce the self-dual and anti-self-dual orbital generators
\begin{equation}
J^{+}_a=\tfrac14\eta_{a\mu\nu}L_{\mu\nu},
\qquad
J^{-}_a=\tfrac14\bar\eta_{a\mu\nu}L_{\mu\nu},
\label{eq:so4-split}
\end{equation}
with $L_{\mu\nu}=-i(x_\mu\partial_\nu-x_\nu\partial_\mu)$. They generate the two commuting factors in the $SU(2)_+\times SU(2)_-$ decomposition of the orbital algebra. For adjoint color generators $(T_a)_{bc}=-i\varepsilon_{abc}$, the Faddeev--Popov operator takes the angular normal form
\begin{equation}
\varepsilon^{abc}\eta_{c\mu\nu}x_\nu\partial_\mu
=2\bigl(\bm T\!\cdot\!\bm J^{+}\bigr)_{ab},
\qquad
\mathcal M=-\partial^2+2\hh(r)\,\bm T\!\cdot\!\bm J^{+}.
\label{eq:normal-form-4d}
\end{equation}
Hence $\mathcal M$ preserves $(\bm J^{+})^2$, $(\bm J^{-})^2$, and $\bm K^2$, where $\bm K=\bm T+\bm J^{+}$.

A hyperspherical harmonic of degree $n$ carries $(j_+,j_-)=(n/2,n/2)$. In a channel of total $K$, the color--orbital coupling is
\begin{equation}
c_{Kj_+}
=\bm T\!\cdot\!\bm J^{+}
=\tfrac12\bigl[K(K+1)-j_+(j_++1)-2\bigr].
\label{eq:cK}
\end{equation}
For $j_+\geq1$, the channels $K=j_++1,j_+,j_+-1$ have $c=j_+,-1,-(j_++1)$. Degree $n=1$ is exceptional because $j_+=1/2$ does not admit the last member. The allowed couplings are therefore
\begin{equation}
K=\tfrac12,\quad 2c=-2,
\qquad
K=\tfrac32,\quad 2c=+1,
\label{eq:n1-allowed-4d}
\end{equation}
and the formally continued value $2c=-3$ is absent. This missing channel will reverse the positive-branch threshold ordering below.

Writing $\psi=f(r)Y_n$ and $u=r^{3/2}f$ gives the reduced radial equation
\begin{equation}
-u''+\frac{\ell(\ell+1)}{r^2}u+2c\,\hh(r)u=\lambda u,
\qquad
\ell=n+\tfrac12.
\label{eq:radial-4d}
\end{equation}
Along an amplitude ray, the perturbation on every fixed angular sector is a bounded multiplication operator linear in the amplitude. The angular decomposition is therefore preserved and the regularity assumptions used in \cref{prop:visibility} hold locally at each critical amplitude.

The source has the form
\begin{equation}
\bigl(T_Ae_{\mu d}\bigr)^a
=r\hh(r)\,N[\mu,d]^a{}_{\nu}\hat x_\nu,
\qquad
N[\mu,d]^a{}_{\nu}=\varepsilon^{a\ell d}\eta_{\ell\mu\nu}.
\label{eq:source-4d}
\end{equation}
Every source vector therefore has hyperspherical degree $n=1$, whereas the radial equation determines which degree reaches the stability boundary. For smooth Dirichlet profiles and compactly supported profiles on $\R^4$, the source is admissible and $V_A=T_AT_A^*$ is well defined. The regular-gauge BPST profile fails the corresponding $L^2$ condition and will be treated separately at the quadratic-form level.

\begin{theorem}[Four-dimensional source-selection rule]
\label{thm:selection-4d}
Let $h$ be differentiable and not identically zero, and assume the source is admissible in the sense of~\eqref{eq:source-admissibility}. For the background~\eqref{eq:thooft-hedgehog},
\begin{equation}
P_nV_AP_n=0
\qquad
\text{for every }n\neq1.
\label{eq:selection-4d}
\end{equation}
Thus every isolated crossing with $n\neq1$ is uniformly dark along the radial amplitude ray. The visibility operator has rank $12$, while in degree one its angular eigenvalues and total weights are
\begin{equation}
\kappa_{1/2}=4,\quad w_{1/2}=16,
\qquad
\kappa_{3/2}=1,\quad w_{3/2}=8.
\label{eq:weights-4d}
\end{equation}
Both allowed degree-one channels are fully visible on a Dirichlet ball.
\end{theorem}

\begin{proof}
Equation~\eqref{eq:source-4d} places $\operatorname{ran}T_A$ entirely in degree one, while \eqref{eq:normal-form-4d} preserves each hyperspherical degree under $\mathcal M$ and its spectral projectors. Orthogonality of distinct degrees immediately gives \eqref{eq:selection-4d}.

At $n=1$, the orbital representation is $(j_+,j_-)=(\tfrac12,\tfrac12)$. Since color couples only to $j_+$,
\begin{equation}
1\otimes\tfrac12=\tfrac12\oplus\tfrac32,
\qquad
\dim\mathcal H_{K=1/2}=4,
\qquad
\dim\mathcal H_{K=3/2}=8.
\label{eq:degree-one-decomposition}
\end{equation}
Defining $C=2\bm T\!\cdot\!\bm J^+$, the identity $C=\bm K^2-\bm T^2-(\bm J^+)^2$ gives eigenvalues $-2$ and $+1$ on the $K=1/2$ and $K=3/2$ multiplets, respectively.

After removing the common radial factor from \eqref{eq:source-4d}, the angular source kernel is
\begin{equation}
(\mathcal K_{\rm src})_{a\nu,b\sigma}
=\sum_{\mu,d}N[\mu,d]^a{}_{\nu}N[\mu,d]^b{}_{\sigma}.
\label{eq:angular-source-kernel-4d}
\end{equation}
The standard contractions of $\varepsilon_{abc}$ and the 't~Hooft symbols reduce this kernel to
\begin{equation}
(\mathcal K_{\rm src})_{a\nu,b\sigma}
=2\delta_{ab}\delta_{\nu\sigma}+\varepsilon_{abc}\eta_{c\nu\sigma}.
\label{eq:angular-source-contraction-4d}
\end{equation}
On the Cartesian degree-one basis this is equivalently
\begin{equation}
{\mathcal K_{\rm src}=2I-C.}
\label{eq:source-identity-4d}
\end{equation}
Since $C=-2$ on $K=1/2$ and $C=+1$ on $K=3/2$,
\[
\mathcal K_{\rm src}=4P_{1/2}+P_{3/2}.
\]
The eigenvalues are therefore $4$ and $1$, and multiplication by the corresponding dimensions $4$ and $8$ gives the weights in \eqref{eq:weights-4d}. In particular, $\rank\mathcal K_{\rm src}=12$.

It remains to exclude a vanishing radial overlap in the visible channels. For a zero mode on a ball of radius $R$, multiplication of \eqref{eq:radial-4d} by $r^{5/2}$ and integration by parts gives
\begin{equation}
2c\,\langle\rho,u\rangle=R^{5/2}u'(R),
\qquad
\rho(r)=r^{5/2}\hh(r).
\label{eq:virial-4d}
\end{equation}
A nontrivial Dirichlet solution cannot satisfy both $u(R)=0$ and $u'(R)=0$, so the overlap is nonzero. Both degree-one multiplets are therefore fully visible.
\end{proof}

Using
\[
\int_{S^3}\hat x_\nu\hat x_\sigma\,\d\Omega
=\frac{\pi^2}{2}\delta_{\nu\sigma},
\]
the source-sector trace reduces to the two visible one-dimensional resolvents,
\begin{equation}
\Hor(A)=\frac{\pi^2}{2}\left[
16\langle\rho,h_{1/2}^{-1}\rho\rangle
+8\langle\rho,h_{3/2}^{-1}\rho\rangle
\right],
\label{eq:horizon-reduced-4d}
\end{equation}
with
\[
h_K=-\partial_r^2+\frac{15}{4r^2}+2c_K\hh(r).
\]
For a normalized state the pole coefficient is weighted by $\kappa_K$, whereas the trace over an entire critical multiplet carries $w_K$.

\subsection{Dark crossings on both amplitude branches}

Let $\hh=\hg\varphi$ with $\varphi\geq0$ and define the radial threshold functional
\begin{equation}
\mu_n^{(4)}[\varphi]
=\inf_{u\neq0}
\frac{\displaystyle\int_0^\infty
\left(|u'|^2+\frac{(n+\frac12)(n+\frac32)}{r^2}u^2\right)\d r}
{\displaystyle\int_0^\infty\varphi(r)u(r)^2\d r}.
\label{eq:mu-4d}
\end{equation}
On the negative branch the attractive channel has $2c=n$, and hence
\begin{equation}
|\hg_\ast^-(n)|=\frac{\mu_n^{(4)}[\varphi]}{n}.
\label{eq:negative-4d}
\end{equation}
The positive branch must retain the degree-one exception,
\begin{equation}
\hg_\ast^+(1)=\frac{\mu_1^{(4)}[\varphi]}{2},
\qquad
\hg_\ast^+(n)=\frac{\mu_n^{(4)}[\varphi]}{n+2}
\quad(n\geq2).
\label{eq:positive-4d}
\end{equation}

For a unit delta shell at $r_0$, the same matching argument as in \cref{lem:delta-shell} gives
\begin{equation}
\mu_n^{(4)}=\frac{2n+2}{r_0}.
\label{eq:delta-mu-4d}
\end{equation}
The two branches then become
\begin{equation}
|\hg_\ast^-(n)|
=\frac{2}{r_0}\left(1+\frac1n\right),
\qquad
\hg_\ast^+(1)=\frac2{r_0},
\qquad
\hg_\ast^+(n)=\frac{2(n+1)}{(n+2)r_0}\quad(n\geq2).
\label{eq:branches-4d}
\end{equation}
The negative thresholds decrease toward the unattained limit $2/r_0$ through dark degrees. On the positive branch, however, the minimum occurs already at $n=2$, with value $3/(2r_0)$. Thus the first positive crossing is also dark, a direct consequence of the missing degree-one coupling $2c=-3$.

The same ordering survives for a nonempty class of smooth finite-width profiles.

\begin{theorem}[Finite-width source-dark first crossings in four dimensions]
\label{thm:dark-4d}
Let $\varphi_w\geq0$ be bounded, normalized by $\int\varphi_w\d r=1$, and supported in $[r_0-w,r_0+w]$. The minimizing degree is finite on each amplitude branch. Both minimizers satisfy $n\geq2$ whenever
\begin{equation}
\frac{w}{r_0}
<
\left(
\frac{2-\sqrt3}{2\sqrt3+2\sqrt6}
\right)^2
\simeq 1.0\times10^{-3}.
\label{eq:width-4d}
\end{equation}
Both first crossings are then uniformly dark, and compact support makes the source admissible in the sense of~\eqref{eq:source-admissibility}.
\end{theorem}

\begin{proof}
The finite-width estimate of \cref{prop:width} depends only on the centrifugal coefficient and therefore applies to the half-integer effective angular momentum in four dimensions. For $s=\sqrt{w/r_0}$ and the delta-shell values $\mu_1^{(4)}=4/r_0$ and $\mu_2^{(4)}=6/r_0$, it gives
\begin{equation}
\frac{\mu_2^{(4)}[\varphi_w]}{2}
\leq
\frac{3}{r_0(1-\sqrt6\,s)^2},
\qquad
\mu_1^{(4)}[\varphi_w]
\geq
\frac{4}{r_0(1+2s)^2}.
\label{eq:width-bracket-4d}
\end{equation}
Condition~\eqref{eq:width-4d} makes the first bound smaller than the second, so
\begin{equation}
\frac{\mu_2^{(4)}[\varphi_w]}{2}
<
\mu_1^{(4)}[\varphi_w].
\label{eq:exclude-n1-4d}
\end{equation}
This excludes $n=1$ on the negative branch, and after dividing the two sides by two it also excludes $n=1$ on the positive branch. Finally, the centrifugal contribution grows quadratically with $n$, while the attractive color factor grows only linearly. The endpoint sequences therefore diverge for large $n$, so each minimum is attained at a finite degree. By \eqref{eq:selection-4d}, every possible minimizer is dark.
\end{proof}

The same centrifugal estimate bounds the degrees that need to be included in a numerical search.

\begin{corollary}[Finite hyperspherical range]
\label{cor:nmax}
Let $\varphi_w$ be as in \cref{thm:dark-4d}, set $K=\|\varphi_w\|_\infty(r_0+w)^2$, and let $m$ be any endpoint value attained on the branch under consideration. Every minimizer satisfies
\begin{equation}
n_\ast<Km-2
\quad\text{on the negative branch},
\qquad
n_\ast<Km
\quad\text{on the positive branch}.
\label{eq:nmax}
\end{equation}
Thus only finitely many hyperspherical degrees can compete for the endpoint at fixed profile.
\end{corollary}

\begin{proof}
A negative value of the quadratic form requires the attractive well to overcome the centrifugal term somewhere on the support. With $\ell=n+\tfrac12$, this implies
\begin{equation}
\mu_n^{(4)}[\varphi_w]\,\|\varphi_w\|_\infty
>
\frac{(n+\frac12)(n+\frac32)}{(r_0+w)^2}.
\label{eq:tail-bound-4d}
\end{equation}
Dividing by the branch-dependent color factor gives the bounds in \eqref{eq:nmax} directly.
\end{proof}

\begin{proposition}[Continuum finiteness for compact four-dimensional profiles]
\label{prop:continuum-finite-4d}
Assume the hypotheses of \cref{thm:dark-4d}, or more generally let the profile be compactly supported in $(0,\infty)$ with a dark endpoint strictly below the visible degree-one endpoint on the same branch. Then the source is admissible and, with the inverse forms defined by the regulator prescription~\eqref{eq:threshold-inverse}, the reduced horizon function~\eqref{eq:horizon-reduced-4d} remains finite at the dark crossing on $\mathbb R^4$.
\end{proposition}

\begin{proof}
Equation~\eqref{eq:horizon-reduced-4d} contains only degree-one resolvents. Strict separation from the visible attractive threshold gives a coercive estimate $Q_K\geq\theta Q_{3/2}$ with $\theta>0$, while the other visible channel is repulsive. Since $\rho$ is compactly supported, the same evaluation estimate used in \cref{prop:continuum-finite} yields $|\langle\rho,v\rangle|^2\leq C Q_{3/2}[v]$. The Legendre representation~\eqref{eq:legendre} then bounds each visible inverse form by $C/\theta$. The critical zero mode lies in a degree absent from the source trace, so the dark crossing does not produce a divergence in \eqref{eq:horizon-reduced-4d}.
\end{proof}

For the normalized $C^\infty$ bump with $r_0=10$ and $w/r_0=0.04$, direct threshold calculations continue the dark ordering well beyond the sufficient interval in \eqref{eq:width-4d}. The positive branch is minimized at $n=2$, while the negative branch is minimized at $n=7$. Their separation from the visible degree-one threshold is about $24\%$ and $37\%$, respectively, as summarized in \cref{tab:shell-4d}.

\begin{table}[H]
\centering
\small
\renewcommand{\arraystretch}{1.15}
\begin{tabular}{lccccc}
\hline
branch & computed minimum & $|\hg_\ast|$ & visible $n=1$ & relative margin & nearest computed competitor\\
\hline
positive & $n=2$ & $0.158$ & $0.207$ & $24\%$ & $n=3$, $0.172$\\
negative & $n=7$ & $0.2628$ & $0.415$ & $37\%$ & $n=8$, $0.2630$\\
\hline
\end{tabular}
\caption{Finite-width thresholds for the four-dimensional shell with $r_0=10$ and $w/r_0=0.04$. The positive and negative minima occur in the source-dark degrees $n=2$ and $n=7$, well outside the sufficient narrow-shell regime of \eqref{eq:width-4d}.}
\label{tab:shell-4d}
\end{table}

Figure~\ref{fig:four-d-dark} shows the same smooth-shell ordering together with the finite-ball BPST convergence discussed next.

\begin{figure}[H]
\centering
\includegraphics[width=\textwidth]{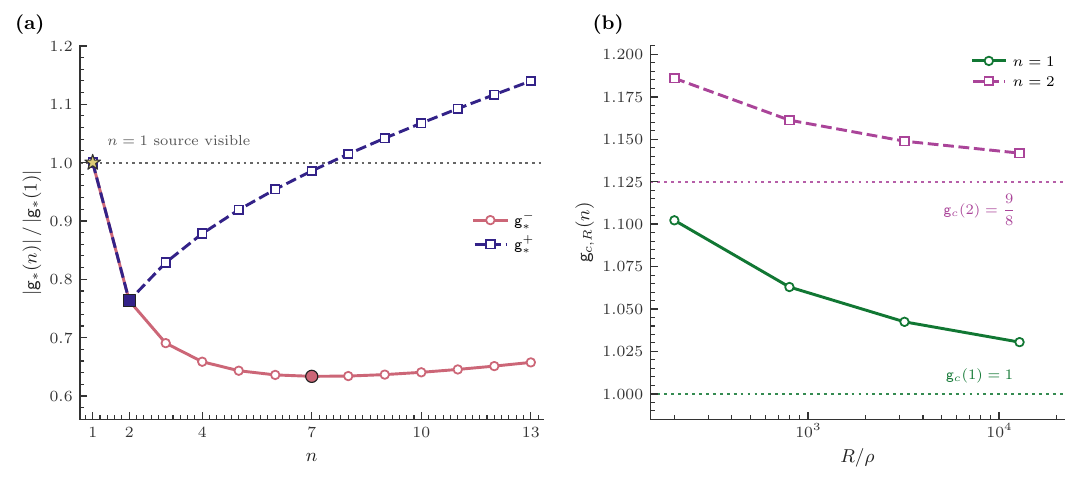}
\caption{Four-dimensional dark ordering and domain check. Panel (a) shows the representative smooth-shell endpoints normalized to the visible $n=1$ value, with dark minima at $n=2$ and $n=7$. Panel (b) shows the finite-ball BPST thresholds approaching the exact half-line values $1$ and $9/8$ from above.}
\label{fig:four-d-dark}
\end{figure}

\section{Positivity thresholds for the regular-gauge BPST background}
\label{sec:bpst}

The regular-gauge BPST background lies at the boundary of the $L^2$ visibility construction. Its inverse-square tail produces generalized zero-energy solutions, while the horizon-source columns are not square integrable on $\R^4$. The appropriate first step is therefore to determine the loss of quadratic-form positivity independently of any full-space horizon trace. In the most attractive positive-branch channel, set
\begin{equation}
a_n=|2c_n|,
\qquad
a_1=2,
\qquad
a_n=n+2\quad(n\geq2).
\end{equation}
After $u=r^{3/2}f$, the radial quadratic form is
\begin{equation}
Q_{n,\hg}^{\rm BPST}[u]
=\int_0^\infty\left[
|u'|^2
+\frac{(n+1)^2-\frac14}{r^2}|u|^2
-\frac{2a_n\hg}{r^2+\rho^2}|u|^2
\right]\d r.
\label{eq:bpst-form}
\end{equation}
Define
\begin{equation}
\delta_n(\hg)=2a_n\hg-(n+1)^2.
\end{equation}
An exact rearrangement gives
\begin{equation}
Q_{n,\hg}^{\rm BPST}[u]
=\mathcal H[u]-\delta_n(\hg)\mathcal W[u]+\mathcal R_{\hg}[u],
\label{eq:bpst-form-decomposition}
\end{equation}
where
\begin{align}
\mathcal H[u]
&=\int_0^\infty\left(|u'|^2-\frac{|u|^2}{4r^2}\right)\d r,\label{eq:hardy-form}\\
\mathcal W[u]
&=\int_0^\infty\frac{|u|^2}{r^2}\,\d r,\\
\mathcal R_{\hg}[u]
&=2a_n\hg\rho^2\int_0^\infty
\frac{|u|^2}{r^2(r^2+\rho^2)}\,\d r.
\end{align}
Hardy's inequality makes $\mathcal H$ positive on the Friedrichs form domain. Recombining the terms gives the useful lower bound
\begin{equation}
Q_{n,\hg}^{\rm BPST}[u]
\geq
\int_0^\infty
\frac{
\bigl[(n+1)^2-2a_n\hg\bigr]r^2
+(n+1)^2\rho^2
}{r^2(r^2+\rho^2)}
|u(r)|^2\,\d r.
\label{eq:hardy-instanton}
\end{equation}

\begin{proposition}[Exact BPST positivity threshold]
\label{prop:bpst-threshold}
The Friedrichs realization of \eqref{eq:bpst-form} is nonnegative if and only if
\begin{equation}
\hg\leq\hg_c(n),
\qquad
\hg_c(1)=1,
\qquad
\hg_c(n)=\frac{(n+1)^2}{2(n+2)}\quad(n\geq2).
\label{eq:bpst-thresholds}
\end{equation}
At $\hg=\hg_c(n)$, the spectral infimum is zero but is not an $L^2$ eigenvalue.
\end{proposition}

\begin{proof}
Because the integral defining $\mathcal R_{\hg}$ is positive, the sign of the remainder is the sign of $\hg$. We therefore separate $\hg\leq0$ from $0\leq\hg\leq\hg_c(n)$ when proving positivity.

If $\hg\leq0$, the original form \eqref{eq:bpst-form} is termwise positive and no Hardy rearrangement is needed. Its centrifugal coefficient satisfies $(n+1)^2-\tfrac14>0$ for every $n\geq1$, and the profile term $-2a_n\hg/(r^2+\rho^2)$ is positive when $\hg\leq0$, so all three integrands are pointwise positive.

If $0\leq\hg\leq\hg_c(n)$, then $\delta_n(\hg)\leq0$ and $\mathcal R_{\hg}\geq0$. The decomposition becomes
\begin{equation}
Q_{n,\hg}^{\rm BPST}[u]
=\mathcal H[u]+|\delta_n(\hg)|\mathcal W[u]+\mathcal R_{\hg}[u]
\end{equation}
and every term is positive, with Hardy's inequality controlling the first. This proves positivity throughout the subcritical and critical range.

For the converse direction, let $\delta=\delta_n(\hg)>0$. Then $\hg>\hg_c(n)>0$, so the BPST remainder is positive but must be shown unable to compensate the negative $-\delta\mathcal W$ term. Choose a nonzero $\chi\in C_c^\infty((-1,1))$ and define
\begin{equation}
A_\chi=\int_{\R}|\chi'(s)|^2\d s,
\qquad
B_\chi=\int_{\R}|\chi(s)|^2\d s.
\end{equation}
For $L,T>0$, let
\begin{equation}
u_{L,T}(r)=r^{1/2}\chi\!\left(\frac{\log r-T}{L}\right).
\label{eq:bpst-witness}
\end{equation}
The test function is smooth and supported in the annulus $e^{T-L}<r<e^{T+L}$, which allows its support to be moved arbitrarily far into the asymptotic inverse-square region by increasing $T$. Under $t=\log r$ and $u(r)=r^{1/2}v(t)$, the critical Hardy form becomes the free one-dimensional form,
\begin{equation}
\mathcal H[u]=\int_{\R}|v'(t)|^2\d t,
\qquad
\mathcal W[u]=\int_{\R}|v(t)|^2\d t.
\end{equation}
Consequently,
\begin{equation}
\mathcal H[u_{L,T}]=\frac{A_\chi}{L},
\qquad
\mathcal W[u_{L,T}]=LB_\chi.
\label{eq:bpst-witness-hardy}
\end{equation}
The positive BPST remainder satisfies
\begin{equation}
\mathcal R_{\hg}[u_{L,T}]
=2a_n\hg\rho^2\int_{\R}
\frac{\chi((t-T)/L)^2}{e^{2t}+\rho^2}\,\d t
\leq2a_n\hg\rho^2e^{-2(T-L)}LB_\chi.
\label{eq:bpst-remainder-bound}
\end{equation}
Choose $L$ so that $A_\chi/L<(\delta/2)LB_\chi$, and then choose $T$ so that the right-hand side of \eqref{eq:bpst-remainder-bound} is smaller than $(\delta/2)LB_\chi$. Equations~\eqref{eq:bpst-form-decomposition} and \eqref{eq:bpst-witness-hardy} then give
\begin{equation}
Q_{n,\hg}^{\rm BPST}[u_{L,T}]<0.
\end{equation}
Hence every $\hg>\hg_c(n)$ admits a compactly supported direction on which the quadratic form is negative. This proves necessity of the threshold bound.

At $\hg=\hg_c(n)$ the $\mathcal W$ coefficient vanishes, leaving $\mathcal H+\mathcal R_{\hg_c}$. Equality in the Hardy inequality is not attained in the Friedrichs form domain, so zero is not an $L^2$ eigenvalue. For a fixed $\phi\in C_c^\infty((1,2))$, the normalized dilates $u_R(r)=R^{-1/2}\phi(r/R)$ satisfy $Q_{n,\hg_c}^{\rm BPST}[u_R]=O(R^{-2})$. Hence the spectral infimum is zero.
\end{proof}

For $n\geq2$, the sequence \eqref{eq:bpst-thresholds} starts at $\hg_c(2)=9/8$ and increases strictly. Thus at the physical value $\hg=1$ all $L^2$ quadratic forms are nonnegative. The degree-one form is Hardy-critical, while every higher degree remains below its instability threshold.

Let $\hg_{c,R}(n)$ be the threshold on $(0,R)$ with a Dirichlet condition at $R$. Enlarging the interval enlarges the variational class, so domain monotonicity makes $\hg_{c,R}(n)$ decrease with $R$ and keeps it above the half-line threshold. Conversely, the compactly supported witness in the proof belongs to every sufficiently large ball for each fixed $\hg>\hg_c(n)$. Therefore
\begin{equation}
\lim_{R\to\infty}\hg_{c,R}(n)=\hg_c(n).
\label{eq:bpst-box-limit}
\end{equation}

The exact regular-instanton solutions discussed by Maas are consistent with this threshold analysis once the operator domain is specified~\cite{Maas2006}. In the dimensionless coordinate $x=r/\rho$, two nontrivial radial solutions in our degree convention are
\begin{align}
f_1(x)
&=
\frac{2\bigl[-x^2+(1+x^2)\log(1+x^2)\bigr]}{x^3},
\label{eq:maas-f1}\\
f_2(x)
&=
\frac{3\bigl[x^4+2x^2-2(1+x^2)\log(1+x^2)\bigr]}{x^4}.
\label{eq:maas-f2}
\end{align}
They satisfy
\begin{align}
f_1''+\frac3x f_1'-\frac3{x^2}f_1+\frac4{1+x^2}f_1&=0,
\label{eq:maas-residual-1}\\
f_2''+\frac3x f_2'-\frac8{x^2}f_2+\frac8{1+x^2}f_2&=0.
\label{eq:maas-residual-2}
\end{align}
Near the origin, $f_1\sim x$ and $f_2\sim x^2$. At infinity,
\begin{equation}
f_1(x)\sim\frac{4\log x}{x},
\qquad
f_2(x)\longrightarrow3.
\label{eq:maas-asymptotics}
\end{equation}
Both solutions remain bounded, but neither is square integrable with the four-dimensional radial measure $x^3\d x$. The first belongs to the source-visible degree $n=1$, while the second lies in the source-dark degree $n=2$. Since neither is an isolated $L^2$ eigenvector, \cref{prop:visibility} assigns no pole residue to them.

The horizon source itself also fails the full-space admissibility condition. Up to a constant, its radial amplitude is $x/(1+x^2)$, so its $L^2$ density is
\begin{equation}
x^3\left(\frac{x}{1+x^2}\right)^2
=\frac{x^5}{(1+x^2)^2}\sim x.
\label{eq:bpst-source-nonl2}
\end{equation}
The density is not integrable at infinity, so the source columns do not belong to $L^2(\R^4)$. Consequently, $T_A$ is not a bounded map into the standard ghost Hilbert space and the full-space trace $\Tr(\mathcal M^{-1}V_A)$ is outside the framework used for compact profiles. A horizon trace for the regular-gauge BPST field requires an additional finite-volume, weighted-space, or density prescription.

Finite volume and full space therefore assign different spectral status to the physical BPST amplitude. On Dirichlet balls, the degree-one and degree-two thresholds decrease toward the half-line values $1$ and $9/8$, so the configuration with $\hg=1$ remains inside every finite-ball region and reaches the Hardy threshold only as $R\to\infty$. Maas's broader boundary convention uses bounded generalized solutions and places the instanton on the first Gribov horizon and on the common boundary with the fundamental modular region~\cite{Maas2006}. The two statements refer to different operator domains. In the Friedrichs quadratic-form setting used here, instability for $\hg>1$ begins in the source-visible degree-one channel.

\section{Conclusion}
\label{sec:follows}

The first loss of Faddeev-Popov positivity and the singular response of Zwanziger's horizon function are distinct spectral questions. The eigenvalues of $\mathcal M[A]$ determine when a channel becomes critical, while the compressed operator $P_cV_AP_c$, with $V_A=T_AT_A^*$, determines whether the horizon source probes that critical subspace. Its rank therefore provides a degeneracy-safe classification of dark, partially visible, and fully visible crossings. In particular, a zero mode can occur at the first loss of positivity without producing a pole in the source-sandwiched inverse when the critical direction is annihilated by the source.

The hedgehog families make this separation explicit. In three dimensions, the source is confined to $L=1$, whereas the radial profile can place the first negative-branch threshold in a channel with $L\geq2$. In four dimensions, the source occupies only hyperspherical degree $n=1$, and the exceptional degree-one angular spectrum permits dark first crossings on both amplitude branches. The finite-width estimates establish smooth profiles with these orderings, while the threshold calculations show that the same mechanism persists far beyond the narrow ranges supplied by the sufficient bounds. The regular-gauge BPST field exposes a different limitation because, at its physical amplitude, the Friedrichs form is Hardy-critical in degree one, but neither the bounded zero-energy solutions nor the horizon-source columns belong to the full-space $L^2$ setting used for compact profiles. A full-space horizon trace for that background consequently requires an additional domain or regularization prescription.
 
These conclusions are configurationwise and do not modify the ensemble horizon condition, the no-pole formulation, or the GZ/RGZ measure. They instead identify spectral proximity to the Gribov boundary and source visibility as quantities that can be examined separately. In the symmetric families considered here, angular selection can suppress the coupling of the horizon source to the critical sector exactly. How this suppression is modified by perturbations that break the underlying symmetry remains an open question.

\subsection*{Declaration on the use of generative AI}
Generative AI tools were used for language editing, algebraic cross-checks, and scripting assistance. All analytical claims, symbolic calculations, and numerical outputs were reviewed by the author, who assumes full responsibility for the manuscript.

\end{document}